\documentclass[onecolumn,draftclsnofoot,12pt]{IEEEtran}
\usepackage{amsmath,amsfonts}
\usepackage{algorithmic}
\usepackage{algorithm}
\usepackage{amsthm}
\usepackage{mathtools}
\usepackage{tikz}
\usetikzlibrary{positioning,shapes,arrows.meta}
\usepackage{amssymb}
\usepackage{array}
\usepackage[caption=false,font=normalsize,labelfont=sf,textfont=sf]{subfig}
\usepackage{textcomp}
\usepackage{stfloats}
\usepackage{url}
\usepackage{verbatim}
\usepackage{hyperref}
\usepackage{graphicx}
\usepackage{cite}
\usepackage{xcolor}
\usepackage{lineno}

\newcommand{\F}{\mathbb{F}}
\renewcommand{\L}{\mathcal{L}}
\newcommand{\U}{\mathcal{U}}
\newcommand{\A}{\mathcal{A}}
\newcommand{\rk}{\mathrm{rk}}

\newcommand{\C}{\mathcal{C}}
\newcommand{\colsp}{\operatorname{colsp}}

\renewcommand{\epsilon}{\varepsilon}
\renewcommand{\P}{\mathcal{P}}

\newcommand{\D}{\mathcal{D}}
\renewcommand{\S}{\mathbf{S}}
\newcommand{\M}{\mathcal{M}}

\newcommand{\tr}{\mathrm{tr}}

\allowdisplaybreaks

\DeclareMathOperator{\rowsp}{rowsp}

\newif\ifcomment
\commenttrue
\renewcommand{\c}[1]{\ifcomment \textcolor{black}{#1}\fi}

\theoremstyle{plain} %\numberwithin{equation}{section}
\newtheorem{theorem}{Theorem}
\newtheorem{corollary}[theorem]{Corollary}

\newtheorem{lemma}[theorem]{Lemma}
\newtheorem{proposition}[theorem]{Proposition}

\theoremstyle{definition}
\newtheorem{definition}[theorem]{Definition}

\newtheorem{remark}[theorem]{Remark}
\newtheorem{example}[theorem]{Example}

\hypersetup{final}
\begin{document}

%\linenumbers

\title{A $q$-Polymatroid Framework for Information Leakage in Secure Linear Network Coding}
\author{Eimear Byrne, Johan Vester Dinesen, Ragnar Freij-Hollanti, Camilla Hollanti
\thanks{J.V.~Dinesen, R.~Freij-Hollanti, and  C.~Hollanti are with the Department of Mathematics and Systems Analysis, Aalto University, Espoo, Finland. Eimear Byrne is with the School of Mathematics and Statistics, University College Dublin, Dublin, Ireland.}
\thanks{Preliminary results were published in the Proceedings of the 2025 IEEE International Symposium on Information Theory \cite{dinesen2025secretsharingrankmetric}.}
\thanks{This work has emanated from research conducted with the ﬁnancial support of the European Union MSCA Doctoral Networks (HORIZON-MSCA-2021-DN-01, Project 101072316). C. Hollanti was supported by the Business Finland Co-Innovation Consortium (Grant No. BFRK/473/31/2024, 5845483) and by the TU Munich Global Visiting Professor Program as part of the Excellence Strategy of the Federal and State governments of Germany.}
}


\maketitle

\begin{abstract}
We study information leakage in secure linear network coding schemes based on nested rank-metric codes. We show that the amount of information leaked to an adversary that observes a subset of network links is characterized by the conditional rank function of a representable $q$-polymatroid associated with the underlying rank-metric code pair. Building on this connection, we introduce the notions of $q$-polymatroid ports and $q$-access structures and describe their structural properties. Moreover, we relate minimal codewords in the rank-metric setting to minimal reconstructing spaces and prove a $q$-analogue of the Brickell--Davenport theorem.

\end{abstract}

\section{Introduction}

Network coding enables intermediate nodes in a communication network to transmit linear combinations of incoming messages rather than simply forwarding them. This approach improves throughput, reliability, and robustness in multicast communication~\cite{informationflow,linearnetworkcoding}. However, beyond efficiency, modern communication systems must also ensure confidentiality. Ensuring that network transmissions remain secure even in the presence of eavesdroppers gives rise to the problem of secure network coding.

A classical foundation for this problem is the Wiretap~II channel introduced by Ozarow and Wyner~\cite{OzarowWyner}, where an adversary can observe a limited number of transmitted symbols. In this setting, linear coset codes provide secrecy guarantees by ensuring that each observation reveals no information about the secret message. This model naturally extends to network scenarios in which the eavesdropper intercepts a subset of communication links \cite{securenetwork}. Secure linear network coding builds on this idea. In such constructions, the algebraic structure of rank-metric coset codes ensures universality with respect to the underlying network topology and robustness against rank deficiencies and transmission errors~\cite{4608992,Universal,Emina}.

Beyond their role in secure communication, these coding schemes admit a rich combinatorial interpretation. Wiretap~II schemes based on linear Hamming-metric coset codes are equivalent to \emph{linear secret sharing schemes} (LSSS)~\cite{shamir79,Blakley,10.1007/978-3-642-20901-7_2,JunKURIHARA2012,1397965}. In an LSSS, the information held by each participant corresponds to a share derived from a codeword, and the qualified sets of participants are characterized by the linear dependencies among these shares. Matroids often appear when linear dependence plays a structural role, making them a convenient object to use for several problems in information theory and coding theory \cite{whitney35,Oxley}. { For example, the entropy of a finite set of random variables induces a polymatroid \cite{FUJISHIGE197855}.} Additionally, a linear code determines a representable matroid whose rank function captures which sets of coordinates support non-trivial codewords; conversely, many code parameters and duality properties can be expressed purely in matroidal terms \cite{greene}. This perspective is particularly useful in secret sharing, where qualified and unqualified coalitions correspond to dependent and independent sets in the underlying matroid, and where circuits describe minimal reconstruction sets \cite{farre_padro2007,Cramer_Damgaard_Nielsen_2015}. Related ideas also arise in distributed storage, where locality and repair groups can be modeled through matroid rank constraints, providing a combinatorial way to describe which subsets of nodes suffice for recovery \cite{Freij-Hollanti2018}. Furthermore, every LSSS corresponds to a \emph{port} of a representable matroid; see~\cite{farre_padro2007} for further details. Matroid ports were first introduced by Lehman~\cite{Lehman} in connection with a game-theoretic problem posed by Shannon. In this setting, the secret is identified with a particular element of the matroid's ground set, and the qualified coalitions correspond to circuits containing that element. Brickell and Davenport~\cite{brickell_davenport} showed that every ideal access structure arises from a matroid port. Massey~\cite{Massey1999MinimalCA} later strengthened the connection between coding theory and secret sharing by showing that the minimal qualified sets in an LSSS correspond precisely to certain minimal codewords in the dual code, thus establishing a direct link between the algebraic properties of the code and the combinatorial access structures.

Together, these classical correspondences illustrate how the combinatorial structure of matroids underlies secure communication and secret sharing. Extending this perspective to the rank-metric setting naturally lends itself to the theory of $q$-matroids. These were introduced to the coding theory community by Jurrius and Pellikaan~\cite{JurriusPellikaan2018}, while subsequent work on $q$-polymatroids further established their relationship with rank-metric codes ~\cite{Gorla_2019,DBLP:journals/corr/abs-2104-06570,Shiromoto}. In this framework, the invariants of a rank-metric code, such as generalized rank weights, duality, and the maximum rank distance (MRD) property, admit natural combinatorial interpretations through the associated $q$-polymatroid. Consequently, the ports of representable $q$-polymatroids form a natural framework for studying the information leakage of linear network coding schemes based on nested rank-metric codes, laying the foundation for the developments presented in this work. \c{By translating the algebraic constraints of the codes into the combinatorial framework of $q$-polymatroids, we obtain a setting to study structural bounds and limits of secure network coding, analogous to the classical application of matroid theory in secret sharing.}

\subsection{Contributions and Related Work}
This work develops a $q$-polymatroidal framework for studying information-theoretic security in linear network coding using nested linear rank-metric code pairs. Our main contributions are as follows.
\begin{itemize}
\item We show that the information leakage in secure linear network coding with nested rank-metric coset codes is determined by the conditional rank function of a representable $q$-polymatroid. 
\item We extend Massey’s link between minimal codewords and minimal reconstructing sets to vector rank-metric codes by showing that minimal dual codewords give rise to minimal reconstructing spaces with respect to rank support. We also show that, unlike in the classical setting, the converse need not hold, even when the dual code is minimal.
\item We prove a $q$-analogue of the Brickell–Davenport theorem by introducing ports and access structures for $q$-polymatroids and showing that, with certain assumptions, the restriction of the underlying structure to a suitable hyperplane is a $q$-matroid.
\end{itemize}

\subsection{Organization}
The paper is organized as follows. Section \ref{sec:prelim} provides the necessary preliminaries on $q$-polymatroids, rank-metric codes, linear network coding, and nested coset coding with rank-metric codes. Section \ref{sec:leakage} describes information leakage associated to nested coset codes via $q$-polymatroids, which motivates Section \ref{sec:ports} where we define $q$-polymatroid ports and $q$-access structures and describe their properties. We conclude by summarizing the results and discussing future directions in Section \ref{sec:conclusions}.

\section{Preliminaries} \label{sec:prelim}

Throughout, $q$ denotes a fixed prime power and $\mathbb{F}_q$ denotes the finite field of order $q$. We let $E$ denote a fixed $n$-dimensional vector space over~$\mathbb{F}_q$. Given a vector space $A$ over $\mathbb{F}_q$, we let $\mathcal{L}(A)$ denote the lattice of subspaces of $A$, ordered by inclusion, and we let $\P(A)$ denote the set of one-dimensional subspaces contained in $A$. 
For $V \leq A$, let $V^{\perp A}$ denote the orthogonal complement of $V$ with respect to a fixed non-degenerate symmetric bilinear form $\langle \cdot , \cdot \rangle$ on $A$, i.e
\[
    V^{\perp A} = \{ w\in A \,|\, \langle v,w\rangle = 0 \text{ for all } v \in V\}.
\]
If the ambient space $A$ is clear from context, we simply write $V^\perp$.

\subsection{\texorpdfstring{$q$}{q}-Polymatroids}
In \cite{Gorla_2019,Shiromoto} the authors study $q$-polymatroids as $q$-analogues of polymatroids, thereby extending the classical notion of matroids into the subspace setting. Associated to each rank-metric code is a pair of $q$-polymatroids. It has been shown that key structural invariants such as generalized weights, the MRD property, characterization of optimal anticodes, and MacWilliams duality can be captured through this combinatorial framework. We recall the definitions and results on $q$-polymatroids that will be used throughout. For broader surveys and perspectives, see \cite{DBLP:journals/corr/abs-2104-06570, BYRNE2024105799, vertigan}.

\begin{definition} A (rational) \textit{$q$-polymatroid} is a pair $\mathcal{M} = (E,\rho)$ where $\rho\colon \mathcal{L} (E) \rightarrow \mathbb{Q}_{\geq 0}$ is a \textit{rank function} such that the following axioms hold.
\begin{enumerate}
    \item[(R1)] \textit{Boundedness}: $0\leq \rho(V) \leq \dim V$ for all $V\leq E$.
    \item[(R2)] \textit{Monotonicity}: $\rho(V)\leq \rho(W)$ for all $V,W\leq E$ with $V\leq W$.
    \item[(R3)] \textit{Submodularity}: $\rho(V + W) + \rho(V\cap W ) \leq \rho(V) + \rho(W)$ for all $V,W\leq E$.
\end{enumerate}
The \emph{rank} of $\M$ is $\rho(E)$. \label{def:qpolymatroid}
\end{definition}

With respect to Definition~\ref{def:qpolymatroid}, $q$-matroids are the integral-valued specialization of $q$-polymatroids and play a central role in the applications considered in this work. Specifically, a $q$-polymatroid $\M=(E,\rho)$ such that $\rho\colon \mathcal{L}(E)\rightarrow \mathbb{N}_0$ is called a \textit{$q$-matroid}. Furthermore, for any $q$-polymatroid $\mathcal{M} = (E,\rho)$ and for any $V,W\leq E$ define $$\rho(V\,|\,W) \coloneqq \rho(V+W) - \rho(W)$$ to be the \textit{conditional rank of $V$ given $W$}.

The following definition provides an example of one of the simplest families of $q$-matroids.

\begin{definition} Let $k\in \mathbb{N}$ such that $k\leq n = \dim E$. The $q$-matroid $\U_{n,k}=(E,\rho)$ such that $\rho(V) = \min \{ \dim V, k\}$ is the \textit{uniform $q$-matroid of rank $k$}. \label{def:unif}
\end{definition}

As in the matroid case, there are many axiomatic definitions of $q$-matroids, often referred to as cryptomorphisms (see e.g.\cite{BYRNE2022149}, \cite{gordonmcnulty}). In general, such axiomatic descriptions for $q$-polymatroids are more complicated and require further information. The following definitions are provided specifically for $q$-matroids, which extend the matroid notions to the subspace setting.

\begin{definition} Let $\mathcal{M} = (E,\rho)$ be a $q$-matroid and $V\leq E$. 
\begin{enumerate}
    \item $V$ is an \textit{independent space} of $\M$ if $\rho(V) = \dim V$.
    \item $V$ is a \textit{dependent space} of $\M$ if it is not independent.
    \item $V$ is a \textit{circuit} of $\M$ if it is dependent and all proper subspaces of $V$ are independent.
    \item $V$ is a \textit{basis} of $\M$ if it is independent and is not properly contained in any independent space of $\M$. In particular, $\dim(V)=\rho(V)=\rho(E)$.
\end{enumerate}
\end{definition}

Next, we recall standard operations on $q$-polymatroids, which play a central role in structural results.

\begin{definition}
    Let $\mathcal{M} = (E,\rho)$ be a $q$-polymatroid. Define $\rho^*(V) \coloneqq \dim V - \rho(E) + \rho(V^{\perp})$ for all $V\leq E$. Then $\mathcal{M}^* \coloneqq (E,\rho^*)$ is a $q$-polymatroid called the \textit{dual} of $\mathcal{M}$, and $(\mathcal{M}^*)^* = \mathcal{M}$.
\end{definition}

\begin{definition} \label{def:equivalentmatroids}
Let $\mathcal{M}_i = (E_i, \rho_i)$ be $q$-polymatroids for $i=1,2$. Then $\mathcal{M}_1$ and $\mathcal{M}_2$ are \textit{equivalent}, denoted by $\mathcal{M}_1 \simeq \mathcal{M}_2$, if there exists an $\mathbb{F}_q$-linear isomorphism $\phi \colon E_1 \rightarrow E_2$ such that $\rho_2(\phi(V)) = \rho_1(V)$ for all $V\leq E_1$. 
\end{definition}

\begin{definition} Let $\mathcal{M} = (E,\rho)$ be a $q$-polymatroid and let $Z\leq E$. Then $\mathcal{M}|_Z \coloneqq (Z,\rho|_{Z})$ where $\rho|_{Z}(V) \coloneqq \rho(V)$ for all $V\leq Z$ is a $q$-polymatroid, called the \textit{restriction of $\mathcal{M}$ to $Z$}. Additionally, let $\pi\colon E\rightarrow E/Z$ denote the canonical projection. Then $\mathcal{M}/Z \coloneqq (E/Z,\rho_{E/Z})$ where $\rho_{E/Z}(V) \coloneqq \rho(\pi^{-1}(V))-\rho(Z)$ is a $q$-polymatroid, called the \textit{contraction of $\mathcal{M}$ by $Z$}.
\end{definition}

\subsection{Rank-Metric Codes}

Let us recall some standard definitions and properties of rank-metric codes. For an overview of rank-metric codes, we refer to \cite{bartz2022rankmetriccodesapplications}. 

\begin{definition}
    A \emph{rank-metric code} is an $\mathbb{F}_q$-linear subspace $\mathcal{C} \leq \mathbb{F}_q^{n \times m}$ equipped with the rank distance $\textup{d}_{\rk}(X, Y) = \mathrm{rk}(X - Y)$
    for any $X, Y \in \mathcal{C}$. The \textit{minimum rank distance} of $\mathcal{C} \neq 0$ is defined as $
    \textup{d}_{\rk}(\mathcal{C}) = \min \{ \mathrm{rk}(X) \mid X \in \mathcal{C}, X \neq 0 \}.$
\end{definition}

A rank-metric analogue of the Singleton bound in the Hamming metric also holds.

\begin{proposition} \label{thm:singleton}
    Let $\mathcal{C} \leq \mathbb{F}_q^{n \times m}$ be a non-zero rank-metric code. Then
    $$
    \dim \mathcal{C} \leq \max\{m, n\} \left( \min\{m, n\} - \textup{d}_{\rk}(\mathcal{C}) + 1 \right).$$
\end{proposition}

\begin{definition}
    A rank-metric code $\mathcal{C} \leq \mathbb{F}_q^{n \times m}$ is \emph{maximum rank distance (MRD)} if it meets the bound of Proposition \ref{thm:singleton} with equality.
\end{definition}

\begin{definition}
    The \textit{dual code} of a rank-metric code $\mathcal{C} \leq \mathbb{F}_q^{n \times m}$ is
    $$
    \mathcal{C}^\perp = \left\{ X \in \mathbb{F}_q^{n \times m} \mid \mathrm{tr}(X Y^\top) = 0 \text{ for all } Y \in \mathcal{C} \right\},$$
    where $\mathrm{tr(A)}$ denotes the trace of a matrix $A$.
\end{definition}

From a rank-metric code $\mathcal{C} \leq \mathbb{F}_q^{n \times m}$ we can induce a $q$-polymatroid by considering shortened subcodes of $\mathcal{C}$. 
In particular, for any $V\leq \mathbb{F}_q^n$, the \emph{shortened subcode of 
$\mathcal{C}$ with respect to $V$} is
\[
\mathcal{C}(V) = \left\{ X \in \mathcal{C} \mid \operatorname{colsp}(X) \leq V \right\}.
\]

This construction provides the link between rank-metric codes and $q$-polymatroids, which we now recall.

\begin{proposition}\label{prop:dualpoly}
    Let $\mathcal{C} \leq \mathbb{F}_q^{n \times m}$ be a rank-metric code. For any $V \leq \mathbb{F}_q^n$, define
    \[
    \rho_{\mathcal{C}}(V) \coloneqq \frac{\dim \mathcal{C} - \dim \mathcal{C}(V^\perp)}{m}.
    \]
    Then $\mathcal{M}_{\mathcal{C}} = (\mathbb{F}_q^n, \rho_{\mathcal{C}})$ is a $q$-polymatroid and $\mathcal{M}^*_{\mathcal{C}} = \mathcal{M}_{\mathcal{C}^\perp}$.
\end{proposition}

We will use the following lemma in our subsequent results, which follows from the definition of $\rho^*$ and Proposition \ref{prop:dualpoly} (see also \cite{bbds_2018,rav_2016}). 

\begin{lemma}\label{lem:ccdual}
Let $\mathcal{C} \leq \mathbb{F}_q^{n \times m}$ be a rank-metric code and let $V \leq \mathbb{F}_q^n$.
Then $$\dim \C^\perp(V) = \dim \C^\perp-m\dim V^\perp +\dim \C(V^\perp).$$
\end{lemma}

%\begin{proof}
 %   From the definitions of $\rho_\C$, $(\rho_\C)^*$, as well as Proposition \ref{prop:dualpoly}, we have 
  %  \begin{eqnarray*}
   %     \dim \C^\perp(V) &=& \dim \C^\perp - m \rho_{\C^\perp}(V^\perp)\\
    %    & = & \dim \C^\perp  -m \dim V^\perp+\dim \C -m \rho_\C (V) \\
        %& = & \dim \C^\perp  -m \dim V^\perp+\dim \C -\dim \C +\dim \C(V^\perp)\\
     %   &=& \dim \C^\perp  -m \dim V^\perp +\dim \C(V^\perp).
    %\end{eqnarray*}
%\end{proof}

\begin{definition}
    Let $\mathcal{M} = (E, \rho)$ be a $q$-polymatroid, where $E$ is an $n$-dimensional $\mathbb{F}_q$-vector space. Then $\mathcal{M}$ is called \emph{representable} if there exists a rank-metric code $\mathcal{C} \leq \mathbb{F}_q^{n \times m}$ such that $\mathcal{M} \simeq \mathcal{M}_{\mathcal{C}}$.
\end{definition}

%The $q$-polymatroids induced by rank-metric codes are called representable. 
In \cite{dinesen2025secretsharingrankmetric} it was shown that representable $q$-polymatroids also admit a description via the entropy of discrete random variables. We present this result here with proof for self-containment. We let $H(\mathbf{X}_1,\ldots,\mathbf{X}_n)$ denote the joint entropy of a set of discrete random variables $\mathbf{X}_1,\ldots,\mathbf{X}_n$. All logarithms and entropies are in base $q$ throughout.

\begin{theorem} \label{thm:entropy}
Let $\mathcal{C}\leq \mathbb{F}_q^{n\times m}$ be a rank-metric code. For any $V\leq \mathbb{F}_q^n$ let $\pi_V \colon \mathcal{C}\rightarrow \mathcal{C}/\mathcal{C}(V^{\perp \mathbb{F}_q^n})$ denote the canonical quotient map and $\mathbf{Z}_V$ the random variable on $\mathcal{C}/\mathcal{C}(V^{\perp\mathbb{F}_q^n})$ induced by the uniform distribution on $\mathcal{C}$ and $\pi_V$. Then
\begin{enumerate}
    \item $H(\mathbf{Z}_V)= m\rho_{\mathcal{C}}(V)$ for all $V\leq \F_q^n$,
    \item $H(\mathbf{Z}_{V_1},\ldots,\mathbf{Z}_{V_\ell}) = H( \mathbf{Z}_{V_1+ \cdots + V_\ell})$ for all $V_1,\ldots,V_\ell\leq \F_q^n$,
    \item $H(\mathbf{Z}_{W} \,|\, \mathbf{Z}_{V})=m \rho_{\mathcal{C}}(W\, |\, V)$ for all $V,W\leq \F_q^n$.
\end{enumerate}
\end{theorem}
\begin{proof}
Throughout the proof let $\bot = \bot \mathbb{F}_q^n$. 
Then
    \begin{align*}
     H(\mathbf{Z}_V) &= -\:\:\sum_{\mathclap{X\in \mathcal{C}/\mathcal{C}(V^\perp)}}\:\: \mathbb{P}(X) \log (\mathbb{P}(X)) \\
    &= -\:\:\sum_{\mathclap{X\in \mathcal{C}/\mathcal{C}(V^\perp)}}\:\: \frac{q^{\dim \ker \pi_V}}{q^{\dim (\mathcal{C})}} \log \left( \frac{q^{\dim \ker \pi_V}}{q^{\dim (\mathcal{C})}}\right)\\
    &= -q^{\dim \mathcal{C}/\mathcal{C}(V^\perp)}\frac{q^{\dim \ker \pi_V}}{q^{\dim (\mathcal{C})}} \log \left( \frac{q^{\dim \ker \pi_V}}{q^{\dim (\mathcal{C})}}\right) \\
    &= -q^{\dim \mathcal{C}/\mathcal{C}(V^\perp)}q^{-\dim \mathcal{C}/\mathcal{C}(V^\perp)} \log\left( q^{-\dim \mathcal{C}/\mathcal{C}(V^\perp)}\right) \\
    &= (\dim \mathcal{C} - \dim \mathcal{C}(V^\perp)) \\ 
    &= m\rho_C(V),
\end{align*}
which proves 1). 

For 2), let $\sigma$ be the map $\sigma \colon \mathcal{C} \rightarrow \prod^{\ell}_{i=1}\mathcal{C}/\mathcal{C}(V_i^\perp)$ defined by $\sigma(X)= (\pi_{V_1}(X),\ldots,\pi_{V_\ell}(X))$ for all $X \in \mathcal{C}$. This map is clearly $\mathbb{F}_q$-linear and satisfies $\ker \sigma = \cap^{\ell}_{i=1}\ker \pi_{V_i} = \ker \pi_{\sum_{i=1}^\ell V_i}$. The result then follows by the definition of joint entropy, as $\mathrm{rk}\, \sigma = \mathrm{rk}\, \pi_{\sum^{\ell}_{i=1} V_i}$ and for any $Y=(Y_1,\ldots,Y_\ell)\in \mathrm{im}\, \sigma$, we have that $\mathbb{P}(Y) = q^{\dim \ker \sigma - \dim \mathcal{C}}$. Finally, by the definition of conditional entropy, we have that 3) follows immediately from 1) and 2).\qedhere
\end{proof}

An alternative viewpoint considers $\mathbb{F}_{q^m}$-linear codes, known as vector rank-metric codes.

Let $\Pi = \{\gamma_1, \ldots, \gamma_m\}$ be a basis of $\mathbb{F}_{q^m}$ over $\mathbb{F}_q$. For $v \in \mathbb{F}_{q^m}^n$, let $\Pi(v) \in \mathbb{F}_q^{n \times m}$ denote the matrix whose $(i, j)$-th entry is the $j$-th coordinate of $v_i$ with respect to the basis $\Pi$. The map 
$v \mapsto \Pi(v)$ defines an $\mathbb{F}_q$-linear isomorphism from 
$ \mathbb{F}_{q^m}^n$ to $\mathbb{F}^{n\times m}$. For an 
$\mathbb{F}_{q^m}$ subspace $A \leq \mathbb{F}_{q^m}^n$, define $\Pi(A) = \{ \Pi(v) \mid v \in A \}$. It is well-known that the minimum rank distance of 
    $\Pi(A)$ is independent of the choice of basis $\Pi$. Moreover, for any such basis, we have
    $
    \dim_{\mathbb{F}_q} \Pi(A) = m \cdot \dim_{\mathbb{F}_{q^m}} A.
    $
We define the rank distance between $v,w \in  \mathbb{F}_{q^m}^n$ to be $\textup{d}_\rk(v,w):=\rk(\Pi(v)-\Pi(w))$.

\begin{definition}
    An $\mathbb{F}_{q^m}$-$[n,k]$ \emph{vector rank-metric code} is a 
    $k$-dimensional $\mathbb{F}_{q^m}$-linear subspace $\mathcal{C}$ of 
    $\mathbb{F}_{q^m}^n$. We say that $\mathcal{C}$ is an $\mathbb{F}_{q^m}$-$[n,k,d]$ vector rank-metric code if $\Pi(\mathcal{C})$ has minimum rank distance $d$, where $\Pi$ is a basis of $\mathbb{F}_{q^m}$ over $\mathbb{F}_{q}.$ 
\end{definition}

\subsection{Linear Network Coding and Nested Coset Coding}

Linear network coding has been extensively studied in \cite{informationflow,networkcoding,linearnetworkcoding,algebraicapproach}. Unlike traditional routing, linear network coding enhances both throughput and robustness by allowing intermediate nodes to linearly combine incoming data streams, rather than simply forwarding them. The problem of securing such coded communication against eavesdropping adversaries has been addressed in several works \cite{securenetwork,Universal,Emina}. These studies investigate the conditions under which an adversary, with access to a subset of network links, can be prevented from gaining information about the transmitted messages.

Consider a communication network comprising a set of source nodes and a set of sink nodes. A source node encodes a message $x\in \F_q^{\ell}$ into a matrix $C\in \F_q^{n\times m}$, where $n$ denotes the number of outgoing links from the source. The encoded message is then transmitted through the network, and a sink node receives a matrix $Y = AC \in \F_q^{N\times m}$, where $A\in \F_q^{N\times n}$ is the transfer matrix characterizing the linear transformation induced by the network between the source and the sink. If the transfer matrix $A$ is selected randomly, the scheme corresponds to random linear network coding, as introduced in \cite{1705002}.

Now, suppose that an eavesdropper gains access to $\mu$ links in the network. The adversary then observes a matrix $BC\in \F_q^{\mu \times m}$, where $B\in \F_q^{\mu \times n}$ represents the associated linear combinations of the encoded message. The communication is said to be secure against $\mu$ observations if the observed matrix $BC$ reveals no information about the original message $x$. Moreover, the scheme is termed \emph{universally secure} against $\mu$ observations if this condition holds for all choices of $B\in \F_q^{\mu \times n}$.

\begin{definition}[\c{\cite{7101826}, Definition 7}]
A \emph{coset coding scheme} over the field $\F_q$ with message set $\mathcal{S}$ is a family of disjoint non-empty subsets of $\F_q^{n\times m}$, $\mathcal{P}_\mathcal{S} = \{ \C_{x}\}_{x\in \mathcal{S}}$. The source encodes each $x\in \mathcal{S}$ by choosing $C\in \mathcal{C}_x$ uniformly at random.
\end{definition}

{
\begin{remark}
Although the sets $\C_x$ in the above definition are not required to be cosets, we retain the terminology \emph{coset coding scheme} from \cite{7101826}. In the setting considered below, the sets $\C_x$ will indeed arise as cosets of a linear subcode through the construction from nested linear code pairs.
\end{remark}
}

A \emph{nested linear code pair} is a pair of rank-metric codes $\C_2 \lneq \C_1 \leq \F_q^{n\times m}$. Choose $\mathcal{W}\leq \F_q^{n\times m}$ such that $\C_1 = \C_2 \oplus \mathcal{W}$ and an $\F_q$-linear isomorphism $\psi\colon \F_q^{\ell}\rightarrow \mathcal{W}$, where $\ell = \dim (\C_1/\C_2)$. Thus, a nested linear code pair $\C_2 \lneq \C_1$ defines a coset coding scheme with message set $\F_q^{\ell}\simeq \mathcal{W}$ and $\P_{\mathcal{S}} = \{\C_x\}_{x\in\mathcal{S}}$, where $\C_x = \psi(x) + \C_2$ for $x\in \F_q^\ell$.

By abuse of notation, we let $x$ denote the uniformly distributed message in $\F_q^\ell$ and $BC$ the induced random variable obtained by first drawing $C$ uniformly at random from the coset $\C_x$ and then applying the eavesdropper's linear observation $B$. The quantity $I(x; BC)$ is then the mutual information between these random variables. \c{Recall that a scheme is secure against an adversary observing a specific set of compromised links described by $B$ if the observation reveals no information about the message; in information-theoretic terms, this occurs if and only if $I(x; BC) = 0$. Consequently, the scheme is universally secure against $\mu$ observations if $I(x; BC) = 0$ for all transfer matrices $B \in \F_q^{\mu \times n}$.}

\begin{proposition}[\cite{unifying}, Proposition 4] \label{prop:martinez} Let $\C_2\lneq \C_1 \leq \F_q^{n\times m}$ be a nested linear code pair and $B\in \F_q^{\c{\mu \times n}}$. 
Let $x\in \F_q^{\ell}$ and $C\in \C_x$ be chosen uniformly at random.
Then
\[
    I(x;BC) = \dim \C_2^\perp(\rowsp(B)) - \dim \C_1^\perp(\rowsp(B)).
\]
\end{proposition}

The quantity $\dim \C_2^\perp(\rowsp(B)) - \dim \C_1^\perp(\rowsp(B))$ captures the amount of distinguishability between the cosets $\C_x$, as observed through the adversary's perspective (i.e., the row space of $B$). If this difference is zero for all $B$ of rank at most $\mu$, then the eavesdropper gains no information about the message $x$, ensuring universal security against $\mu$ observations. This property can be guaranteed when the underlying code $\C_2$ is MRD and $\C_1$ is the ambient space \cite{Universal}.

\section{Information Leakage in Nested Coset Coding and \texorpdfstring{$q$}{q}-Polymatroids} \label{sec:leakage}
In this section, we show that any representable $q$-polymatroid gives rise to a nested coset coding scheme for which the information leakage is fully determined by the underlying $q$-polymatroid. Conversely, we establish that information leakage in linear network coding with nested coset coding can be described using $q$-polymatroids induced by the underlying nested pair of rank-metric codes.

In Theorem~\ref{thm:equiv2}, we will use the following construction. Let $\C_2 \lneq \C_1 \leq\F_q^{n\times m}$ be a pair of nested linear codes and let $k_2^\perp = \dim \C_2^\perp$. Fix a basis $\{H_1,\ldots,H_{k_2^\perp}\} \subseteq \F_q^{n\times m}$ of $\C_2^\perp$.  Define the map $\Psi\colon \C_1 \rightarrow \F_q^{(k_2^\perp + n )\times m}$ by
\begin{align} 
    \Psi(X) = \begin{bmatrix} \begin{matrix}
    \mathrm{tr}(H_1 X^\top) & \cdots & \mathrm{tr}(H_1 X^\top)  \\
    \vdots & \ddots & \vdots \\ 
    \mathrm{tr}(H_{k_2^\perp} X^\top) & \cdots & \mathrm{tr}(H_{k_2^\perp} X^\top) 
\end{matrix}  \\ X\end{bmatrix}. \label{eq:paddedcode}
\end{align}
Then for all $X\in \C_1$ we have that
\[
    X\in \C_2 \: \Longleftrightarrow \:\Psi(X) = \begin{bmatrix}
    0 \\ X
\end{bmatrix} \: \Longleftrightarrow \:\colsp( \Psi(X)) \leq \langle \mathbf{e}_{ k_2^\perp + 1},\ldots,\mathbf{e}_{k_2^\perp+n} \rangle_{\F_q} \leq \F_q^{k_2^\perp +n},
\]
\c{where $\mathbf{e}_1,\ldots, \mathbf{e}_{k_2^\perp+n}$ are the standard basis vectors of $\F_q^{k_2^\perp+n}$. We refer to $\Psi$ as a padding operator, since it embeds $\mathcal{C}_1$ into a higher-dimensional ambient space by appending evaluations of the dual basis. For elements in the subcode $\mathcal{C}_2$, this pads the matrix with zero rows, translating algebraic membership into subspace containment constraints.} Finally, observe that $\Psi$ is injective, so the code $\Psi(\C_1) \leq \F_q^{(k_2^\perp + n)\times m}$ satisfies $\dim \Psi(\C_1) = \dim \C_1$.
\c{
\begin{theorem}
\label{thm:equiv2}
Let $\C_2 \lneq \C_1 \leq\F_q^{n\times m}$ be a pair of nested linear codes. Let $x \in \F_q^{\dim(\C_1/\C_2)}$ and let $C \in \C_x$ be an encoding of $x$ via the coset coding scheme defined by $\C_2 \lneq \C_1$, and $B\in \F_q^{\mu \times n}$. Let $k_2^\perp = \dim \C_2^\perp$ and define 
    $$Q_0= \langle \mathbf{e}_1,\ldots,\mathbf{e}_{k_2^\perp}\rangle_{\F_q}\leq \F_q^{k_2^\perp + n},\quad Q = \langle \mathbf{e}_{k_2^\perp + 1},\ldots,\mathbf{e}_{k_2^\perp+n} \rangle_{\F_q} \leq \F_q^{k_2^\perp +n},$$ 
and let $\tau \colon \F_q^n\rightarrow Q$ be the $\F_q$-isomorphism defined by $\tau(\hat{\mathbf{e}}_i) = \mathbf{e}_{k_2^\perp + i}$ for $i=1,\ldots,n$, where $\hat{\mathbf{e}}_1,\ldots,\hat{\mathbf{e}}_n$ are the standard basis vectors of $\F_q^n$. Then
\begin{equation}
    H(x\,|\, BC) =  m\,\rho_{\Psi(\C_1)}(Q_0 \,|\, \tau(\rowsp(B))), \label{eq:portscoset}
\end{equation}
where $\Psi$ is defined as in \eqref{eq:paddedcode}. In particular, if $\C_2 = \C_1(P_0^{\perp \F_q^n})$ for some $P_0\leq \F_q^n$, then
\begin{equation}
    H(x\,|\, BC) =  m \, \rho_{\C_1}(P_0\,|\,\rowsp(B)). \label{eq:portstonested}
\end{equation}
\end{theorem}}
\begin{proof}
\c{We equip $\F_q^{k_2^\perp + n}$ with the standard dot product, and denote orthogonal complements with respect to this form by $\perp = \perp{\F_q^{k_2^\perp + n}}$. With this convention, $Q_0^\perp = Q$. In particular, for any $V \leq Q$, there exists $W \leq Q$ such that $V^\perp = Q_0 \oplus W$. Moreover, since $\tau$ is defined to map the standard basis of $\F_q^n$ to the corresponding standard basis vectors generating $Q$, it preserves the standard dot product. Therefore, we have $\tau^{-1}(W)^{\perp{\F_q^n}} = \tau^{-1}(V)$.} Letting $\D = \Psi(\C_1)$, we obtain 
\[\dim \D(V^\perp) = \dim \D(Q_0 \oplus W ) = \dim \C_1(\tau^{-1}(W)),\]
and similarly 
$$\dim \D((V+Q_0)^\perp) = \dim\D(W) = \dim \C_2 (\tau^{-1}(W)).$$ 
Suppose now that $V\leq Q$ satisfies $\rowsp(B) = \tau^{-1}(V)$. By Proposition \ref{prop:martinez} we then have
\begin{align*}
    H(x \,|\, BC) &= H(x) - I(x;BC)\\ 
    &= \dim \C_1 - \dim \C_2 + \dim \C_1^\perp(\tau^{-1}(V))-  \dim \C_2^\perp(\tau^{-1}(V)) \\
    &= \dim \C_1(\tau^{-1}(V)^{\perp\F_q^n}) - \dim \C_2(\tau^{-1}(V)^{\perp \F_q^n}) \\
    &= \dim \D -\dim \D + \dim \C_1(\tau^{-1}(W)) - \dim \C_2(\tau^{-1}(W)) \\ 
    &= \dim \D - \dim \D((V+Q_0)^\perp) -(\dim \D - \dim \D(V^\perp)) \\
    &= m \rho_{\D}(Q_0\,|\,V),
\end{align*}
where the third equation in the above follows from Lemma \ref{lem:ccdual}.
This proves the first claim. Suppose now also that $\C_2= \C_1(P_0^{\perp \F_q^n})$ for some $P_0\leq \F_q^n$. Similarly, we get
\begin{align*}
    H(x \,|\, BC) &= \dim \C_1(\tau^{-1}(V)^{\perp\F_q^n}) - \dim \C_2(\tau^{-1}(V)^{\perp \F_q^n}) \\
    &=  \dim \C_1(\tau^{-1}(V)^{\perp\F_q^n}) - \dim \C_1(P_0^{\perp \F_q^n} \cap \tau^{-1}(V)^{\perp \F_q^n}) \\
    &= \dim \C_1 - \dim \C_1((P_0+ \tau^{-1}(V))^{\perp \F_q^n}) - (\dim \C_1 - \dim \C_1(\tau^{-1}(V)^{\perp\F_q^n})) \\
    &= m \rho_{\C_1}(P_0 + \tau^{-1}(V)) - m \rho_{\C_1}(\tau^{-1}(V)) \\
    &= m \rho_{\C_1}(P_0\,|\, \tau^{-1}(V)),
\end{align*}
which completes the proof.
\end{proof}

Equation \eqref{eq:portstonested} shows that information leakage depends only on the conditional rank function of $\C_1$. This gives a direct information-theoretic interpretation analogous to classical linear secret sharing. In particular, consider a rank-metric code $\C \leq \F_q^{n\times m}$ and a subspace $P_0 \leq \F_q^n$ with $\rho_\C(P_0)>0$. Let $G_{P_0}$ be a full-rank matrix whose row space is $P_0$, and select some $X \in \C$ such that $G_{P_0}X$ represents the secret. The condition $\rho_\C(P_0)>0$ ensures that at least two codewords in $\C$ yield different secrets, which guarantees ambiguity.

An adversary observing $\mu$ network links obtains $BX$, where $B \in \F_q^{\mu\times n}$ describes the transfer matrix of the compromised links. The adversary can reconstruct the secret $G_{P_0}X$ if and only if, for all $Y_1,Y_2 \in \C$, 
\[
BY_1 = BY_2 \;\Rightarrow\; G_{P_0}Y_1 = G_{P_0}Y_2,
\]
that is, whenever two codewords are indistinguishable through $B$, they must also be indistinguishable through $G_{P_0}$. Because $\C$ is linear, this is equivalent to requiring that $BY = 0$ implies $G_{P_0}Y = 0$ for all $Y \in \C$. The condition $BY=0$ means precisely that $\colsp(Y) \leq \rowsp(B)^\perp$. Hence, the adversary can reconstruct the secret if and only if
\[
\C(\rowsp(B)^\perp) = \C(\rowsp(B)^\perp \cap P_0^\perp),
\]
which is exactly the condition $\rho_\C(P_0 \,|\, \rowsp(B)) = 0$. \c{Conversely, the communication is perfectly secure, meaning it leaks no information, if the adversary's observation yields zero information about the secret message $x$. In information-theoretic terms, this means $H(x \,|\, BC) = H(x)$. Since $H(x) = H(x \,|\, 0) = m \rho_\C(P_0 \,|\, 0) = m \rho_\C(P_0)$, Theorem \ref{thm:equiv2} yields that this perfect privacy condition is equivalent to $m\rho_\C(P_0 \,|\, \rowsp(B)) = m\rho_\C(P_0)$.} Therefore, the conditional rank function $\rho_\C(P_0 \,|\, V)$ quantifies, in rank-metric terms, how much additional information about the secret subspace $P_0$ becomes accessible once the adversary knows the observations from $V$.

\section{\texorpdfstring{$q$}{q}-Access Structures and \texorpdfstring{$q$}{q}-Polymatroid Ports} \label{sec:ports}

Motivated by the correspondence established in Section~\ref{sec:leakage}, we now define $q$-polymatroid ports in direct analogy with matroid ports.

\begin{definition}
Let $\mathcal{M}=(E,\rho)$ be a $q$-polymatroid and $P_0,P \leq E$ such that $P_0 \oplus P = E$ and $\rho(P_0) > 0$. We define $\mathbf{S}_{P_0,P}(\mathcal{M}) = (\Gamma, \mathcal{A})$, where
\begin{align*}
    \Gamma &\coloneq \{V\leq P \, |\, \rho(P_0 \,|\, V ) = 0\}, \\
    \mathcal{A} &\coloneqq \{W\leq P \,|\, \rho(P_0\,|\, W ) = \rho(P_0)\}.
\end{align*}
$\mathbf{S}_{P_0,P}(\mathcal{M})$ is called a \textit{generalized $q$-polymatroid port}. If $ \mathcal{M}$ is a $q$-matroid, then $\mathbf{S}_{P_0,P}(\mathcal{M})$ is a \textit{generalized $q$-matroid port}. If $\dim P_0 = 1$, we call it a \textit{$q$-polymatroid port}. If both conditions are satisfied, it is a \textit{$q$-matroid port}. 
\end{definition}

{
To formally bridge the combinatorial sets $\Gamma$ and $\A$ with the physical security of the network, we explicitly define the extreme cases of information leakage.

\begin{definition}
Let $\C_2 \lneq \C_1 \leq \F_q^{n\times m}$ be a nested linear code pair that defines a coset coding scheme with uniformly distributed message $x$, and let $B \in \F_q^{\mu \times n}$ be an adversary transfer matrix. We say the observation $BC$ \emph{leaks $\lambda$ units of information} if
\[
I(x;BC) = H(x) - H(x \mid BC) = \lambda,
\]
for some $0 \leq \lambda \leq H(x)$. In particular:
\begin{enumerate}
\item If $\lambda = 0$, we say that the observation \emph{leaks no information}.
\item If $\lambda = H(x)$, we say that the observation \emph{leaks all information}.
\end{enumerate} \label{def:leakage}
\end{definition}

By Theorem~\ref{thm:equiv2}, the general coset coding scheme induces a generalized $q$-polymatroid port $\S_{Q_0,Q}(\M_{\Psi(\C_1)}) = (\Gamma, \A)$. In this combinatorial setting, the secret corresponds to the padded subspace $Q_0 \leq \F_q^{k_2^\perp + n}$ and the observations reside in the complement $Q$. For an adversary matrix $B$, let $V = \tau(\rowsp(B)) \leq Q$ be the translated observation space. By Theorem~\ref{thm:equiv2}, the unconditional entropy is $H(x) = H(x \mid 0) = m\rho_{\Psi(\C_1)}(Q_0)$, while \eqref{eq:portscoset} gives $H(x \mid BC) = m\rho_{\Psi(\C_1)}(Q_0 \mid V)$. Hence, in the terminology of Definition~\ref{def:leakage}, the observation leaks $\lambda$ units of information if and only if
$$
\lambda = m\left(\rho_{\Psi(\C_1)}(Q_0) - \rho_{\Psi(\C_1)}(Q_0 \mid V)\right).
$$
In particular, the observation leaks all information, i.e., $\lambda = H(x)$, if and only if $\rho_{\Psi(\C_1)}(Q_0 \mid V) = 0$, that is, $V$ completely determines the secret subspace $Q_0$. This characterizes the reconstructing structure: the observation leaks all information if and only if $V \in \Gamma$. Similarly, the observation leaks no information, i.e., $\lambda = 0$, if and only if $\rho_{\Psi(\C_1)}(Q_0 \mid V) = \rho_{\Psi(\C_1)}(Q_0)$. This means $V$ provides zero dimension-reduction to the secret space, which characterizes the privacy structure: the observation leaks no information if and only if $V \in \A$.

We distinguish this from the specialized case where $\C_2 = \C_1(P_0^{\perp{\F_q^n}})$. In this setting, information leakage is characterized directly by the underlying code $\C_1$ without requiring the padding construction $\Psi$. The relevant port simplifies to $\S_{P_0, P}(\M_{\C_1}) = (\Gamma, \A)$, where $P_0 \leq \F_q^n$ represents the secret subspace, and $P$ is a chosen complement of $P_0$ in $\F_q^n$. For an adversary observation space $V = \rowsp(B) \leq P$, it follows directly from \eqref{eq:portstonested} that the observation leaks $\lambda$ units of information if and only if
$$
\lambda = m\left(\rho_{\C_1}(P_0) - \rho_{\C_1}(P_0 \mid V)\right).
$$
In particular, $V$ leaks all information if and only if $V \in \Gamma$, and $V$ leaks no information if and only if $V \in \A$.

Therefore, generalized $q$-polymatroid ports provide a combinatorial dictionary for information leakage in linear network coding schemes defined by nested linear code pairs. Motivated by this physical interpretation, even when an abstract $q$-polymatroid $\M$ is not representable, we treat its parameters in information-theoretic terms. The rank of the secret space, $\rho(P_0)$, dictates the normalized entropy or the \emph{size of the secret message}. For any one-dimensional subspace $p \in \mathcal{P}(P)$, we refer to $\rho(p)$ as the \emph{size of a single packet} traversing a network link. Furthermore, the conditional rank $\rho(P_0 \mid V)$ captures the remaining normalized uncertainty about the secret given the network observations $V$, while the difference $\rho(P_0) - \rho(P_0 \mid V)$ captures the normalized information leakage caused by the observation.

The preceding interpretation also allows us to distinguish different levels of partial information leakage, as is commonly done in ramp secret sharing. In the literature, non-perfect secret sharing is often analyzed by defining families of share sets, denoted $\A_\ell$, that leak exactly $\ell$ symbols of the secret \cite{IWAMOTO200652}. In our framework, if an observation leaks $\lambda$ units of information, then
$$
\lambda = I(x; BC) = H(x) - H(x \mid BC) = m (\rho(P_0) - \rho(P_0 \mid V)).
$$
We therefore define the normalized information leakage to be
$$
r = \frac{\lambda}{m} = \rho(P_0) - \rho(P_0 \mid V).
$$
Because the rank function of a representable $q$-polymatroid takes values in $\frac{1}{m}\mathbb{Z}$, the possible values of $r$ lie in $\frac{1}{m}\mathbb{Z}$. For $r \in \frac{1}{m}\mathbb{Z}$ such that $0 < r < \rho(P_0)$, we introduce the intermediate access families
$$
\A_r = \left\{ V \le P \;\middle|\; \rho(P_0) - \rho(P_0 \mid V) = r \right\}.
$$
Thus, $\A_r$ consists precisely of the observation spaces that leak $\lambda = mr$ units of information. Moreover, extending this notation to the extremal values gives $\A_0=\A$ and $\A_{\rho (P_0)}=\Gamma$. In this way, the parameter $r$ serves as the normalized $q$-analogue of the discrete leakage parameter $\ell$ from \cite{IWAMOTO200652}, and the families $\A_r$ stratify the observation spaces according to the amount of information they reveal about the secret.

The combinatorial parameters governing these intermediate families $\mathcal{A}_r$ are connected to the weight hierarchies of the underlying codes. The relationship between $q$-matroids and generalized rank weights (GRWs) is established in \cite{Gorla_2019, Britz2020}. It was then shown in \cite{Johnsen_2020} that the generalized weights of flags of codes are completely determined by the rank functions of their associated $q$-polymatroids. For a flag of length two, which corresponds to a nested code pair, these parameters coincide with the relative generalized matrix weights (RGMWs) introduced in \cite{MPM}. For codes linear over an extension field, this formulation captures the relative generalized rank weights (RGRWs) \cite{KMU}. If $r = t/m$ for some integer $t$, then the $t$th RGMW is defined by the minimum dimension of a support space $V \leq \mathbb{F}_q^n$ required to achieve a dimension difference of at least $t$ between the corresponding subcodes $\mathcal{C}_1(V)$ and $\mathcal{C}_2(V)$, directly mirroring the conditional rank evaluations that define our leakage families. To link this to information leakage, \cite[Theorem 1]{MPM} establishes that if an adversary successfully compromises $r$ normalized units of the secret message, equivalently $t = mr$ units of information, the minimum number of required network observations $\mu$ is dictated by the $t$th relative generalized matrix weight between the dual codes $\mathcal{C}_2^\perp$ and $\mathcal{C}_1^\perp$. Consequently, the hierarchy of relative weights is encoded by the underlying $q$-polymatroid rank function, characterizing the physical network conditions under which an observation space $V$ transitions between the partial leakage structures $\A_r$.
}

\begin{remark}\c{
While representable $q$-polymatroids form the operational foundation of the coset coding schemes discussed here, \cite{DegenKuehne2025} showed that asymptotically almost all $q$-matroids are not representable over any finite field. Despite this, studying information leakage through the broader abstract framework of $q$-polymatroids remains useful. Much like matroid theory abstracts linear secret sharing to establish bounds, characterize access structures, and prove existence criteria independent of specific field constraints, $q$-polymatroids provide a purely combinatorial lens for secure network coding. This abstraction isolates the structural properties of information leakage from the algebraic constraints of a specific nested code pair, allowing us to leverage the combinatorial geometry of subspace lattices to better understand universal bounds and access structures in rank-metric cryptography.}
\end{remark}

\begin{proposition} \label{prop:portaccess}
Let $\S_{P_0,P}(\M)=(\Gamma,\A)$ be a generalized $q$-polymatroid port. 
\begin{enumerate}
    \item Let $V\leq P$ such that there exists $V'\in \Gamma$ with $V'\leq V$. Then $V\in \Gamma$.
    \item Let $W\leq P$ such that there exists $W'\in \A$ with $W\leq W'$. Then $W\in \A$.
\end{enumerate}
\end{proposition}
\begin{proof}
Let $\c{V'} \in \Gamma$.  Toward a contradiction, suppose that $\c{V} \leq P$ such that $\c{V'} \leq \c{V}$ but $\c{V} \notin \Gamma$. Then, using submodularity of $\rho$, we get that
\[
    \rho(\c{V}) - \rho(\c{V'}+P_0) < \rho(\c{V}+P_0) - \rho(\c{V'}+P_0 ) \leq \rho(\c{V}) - \rho(\c{V'}).
\]
This implies $\rho(P_0 \,|\, \c{V'})>0$, which contradicts the assumption that $\c{V'} \in \Gamma$. This proves 1).
To see that 2) holds, note that again by submodularity of $\rho$ we have that $\rho(P_0) \geq \rho(W+P_0)-\rho(W) \geq \rho(W'+P_0)-\rho(W').$ Since 
$W' \in \A$, we have $\rho(W'+P_0)-\rho(W')=\rho(P_0)$, from which the result follows.
\end{proof}

By Proposition \ref{prop:portaccess}, a generalized $q$-polymatroid port $\mathbf{S}_{P_0,P}(\M)$ forms the $q$-analogue of an access structure on $\L(P)$. We formalize this in the following.

\begin{definition}
    A subset $\Gamma\subseteq \L(E)$ is called \emph{monotone} if $V\in \Gamma$ and $V'\leq E$ with $V\leq V'$ implies $V'\in \Gamma$. Dually we define a subset $\A\subseteq \L(E)$ to be \emph{anti-monotone} if $W\in \A$ and $W'\leq E$ with $W'\leq W$ implies $\c{W'}\in \A$. Two monotone sets, or two anti-monotone sets, $H_i\subseteq \L(E_i)$, $i=1,2$, are called \emph{equivalent} if there exists an $\F_q$-linear isomorphism $\varphi\colon E_1\rightarrow E_2$ such that $\varphi(H_1) = H_2$.
\end{definition}

\c{In order-theoretic language, a monotone subset of $\L(E)$ is an order filter (or up-set) in the poset $\L(E)$ ordered  by inclusion, while an anti-monotone subset is an order ideal (or down-set).}

\begin{definition}
    Let $\Gamma,\A\subseteq \L(E)$ such that $\Gamma \cap \c{\A}=\varnothing$, $\Gamma$ is monotone and $\A$ is anti-monotone. Then $\mathbf{S} =(\Gamma,\A)$ is a $q$-access structure on $\L(E)$.
\end{definition}

\begin{corollary}
A generalized $q$-polymatroid port $\S_{P_0,P}(\M)=(\Gamma,\A)$ is a $q$-access structure on $\L(P)$.
\end{corollary}

For a $q$-access structure $\S=(\Gamma,\A)$ we call $\Gamma$ the \emph{reconstructing structure} of $\S$ and $\A$ the \emph{privacy structure} of $\S$. Elements of $\Gamma$ are called \emph{reconstructing spaces} and elements of $\A$ are called \emph{privacy spaces}. We define two $q$-access structures to be equivalent if their reconstructing structures and privacy structures are equivalent under the same bijection, as formalised in the following.

\begin{definition}
Let $\mathbf{S}_i = (\Gamma_i,\A_i)$ be $q$-access structures on $\L(E_i)$ for $i=1,2$. If there exists an $\F_q$-linear isomorphism $\varphi\colon E_1 \rightarrow E_2$ such that $\varphi(\Gamma_1) = \Gamma_2$ and $\varphi(\A_1) =\A_2$, then $\mathbf{S}_1$ and $\mathbf{S}_2$ are called \textit{equivalent}, and we denote this by $\S_1 \simeq \S_2$.
\end{definition}

We now define minors on monotone and anti-monotone sets. These definitions are then naturally extended to minors on $q$-access structures. The duality between restriction and contraction is given in Theorem \ref{thm:minors}, analogous to the case of $q$-polymatroids \cite{BYRNE2024105799,DBLP:journals/corr/abs-2104-06570}.

\begin{definition} Let $H\subseteq \mathcal{L}(E)$ be a monotone or anti-monotone set and let $Z\leq E$. Then 
\begin{enumerate}
    \item $\overline{H} = \{V \leq E\,|\, V\notin H\}$ is called the \textit{complement} of $H$.
    \item $H^*=\{ V\leq E\,|\, V^{\perp}\in H \}$ is called the \textit{dual} of $H$.
    \item $H|_Z = \{ V \leq Z\,|\, V\in H\}$ is called the \textit{restriction} of $H$ to $Z$.
    \item $H/Z = \{V \leq E/Z \,|\, \pi^{-1}(V)\in H\}$ is called the \textit{contraction} of $H$ by $Z$, where $\pi\colon E\rightarrow E/Z$ is the canonical quotient map.
\end{enumerate}
\end{definition}
The dual of a monotone set is clearly anti-monotone, and vice-versa. Similarly, both monotonicity and anti-monotonicity are preserved under restriction and contraction. Hence, we obtain $q$-access structures from a given one, by dualization, contraction, and restriction.

\begin{definition} Let $\mathbf{S}=(\Gamma,\mathcal{A})$ be a $q$-access structure on $\mathcal{L}(E)$ and $Z\leq E$. Then
\begin{enumerate}
    \item $\mathbf{S}^* = (\mathcal{A}^*,\Gamma^*)$ is called the \textit{dual} $q$-access structure of $\mathbf{S}$ .
    \item $\mathbf{S} |_Z = (\Gamma|_Z,\mathcal{A}|_Z)$ is called the \textit{restriction} of $\mathbf{S}$ to $Z$.
    \item $\mathbf{S}/Z = (\Gamma/Z,\mathcal{A}/Z)$ is called the \textit{contraction} of $\mathbf{S}$ by $Z$.
\end{enumerate}
\end{definition}

Note that in constructing $\mathbf{S}^*$, the duals of $\Gamma$ and $\mathcal{A}$ are obtained with respect to the same \c{non-degenerate symmetric bilinear form}. Clearly, we have $\mathbf{S}^{**}=\mathbf{S}$. Additionally, the equivalence class of $\mathbf{S}^*$ is independent of the choice of non-degenerate symmetric bilinear form.

 As stated before, the class of $q$-access structures is closed under dualization and taking minors. In particular, $\mathbf{S}|_Z$ is a $q$-access structure on $\mathcal{L}(Z)$ and $\mathbf{S}/Z$ is a $q$-access structure on $\mathcal{L}(E/Z)$. Intuitively, the restriction to $Z$ describes how the $q$-access structure changes when information outside $Z$ is removed from the system, while the contraction by $Z$ describes how a $q$-access structure changes when $Z$ is made public.

 \begin{theorem} \label{thm:minors}
Let $\mathbf{S}$ be a $q$-access structure on $\mathcal{L}(E)$. Then
\[
    (\mathbf{S}/Z)^* \simeq \mathbf{S}^*|_{Z^{\perp}} \quad \text{and} \quad(\mathbf{S}| _{Z^{\perp}})^* \simeq \mathbf{S}^*/Z.
\]
\end{theorem}
\begin{proof}
{We prove the first equivalence; the second then follows trivially by biduality (i.e., it is the dual of the first statement). Let $\pi\colon E\rightarrow E/Z$ denote the canonical quotient map. Since the equivalence class of $\S^*$ is independent of the choice of non-degenerate symmetric bilinear form, we equip $E$ with a form such that the restriction to $Z\times Z$ is non-degenerate, meaning $Z\cap Z^\perp = \{0\}$. This yields the orthogonal direct sum decomposition $E=Z\oplus Z^\perp$.

By \cite[Theorem 5.3]{DBLP:journals/corr/abs-2104-06570}, there exists an $\F_q$-linear isomorphism $\phi\colon E/Z \rightarrow Z^\perp$ satisfying $\pi^{-1}(V/Z)= \phi(V/Z) \oplus Z$ for any $V\leq E$ with $Z\leq V$, which implies $\phi(V/Z) = V\cap Z^\perp$. This isomorphism $\phi$ induces the corresponding lattice isomorphism $V\mapsto \phi(V)$ from $\L(E/Z)$ to $\L(Z^\perp)$ with a monotone inverse. We equip $E/Z$ with the non-degenerate symmetric bilinear form obtained by pulling back the form on $Z^\perp$ via $\phi$. With respect to this form, $\phi$ preserves orthogonality, that is, for any $W\leq E/Z$, $\phi(W^{\perp E/Z}) = \phi(W)^\perp \cap Z^\perp$. Let $W\leq E/Z$. By definition, $W\in (\Gamma/Z)^*$ if and only if $W^{\perp E/Z}\in \Gamma/Z$, which means $\pi^{-1}( W^{\perp E/Z})\in \Gamma$. Using the decomposition, the preimage under $\pi$ is
\[
\pi^{-1}(W^{\perp E/Z}) = \phi(W^{\perp E/Z})\oplus Z = (\phi(W)^\perp \cap Z^\perp)\oplus Z = \phi(W)^\perp \cap(Z^\perp \oplus Z)= \phi(W)^\perp,
\]
since $Z\leq \phi(W)^\perp$. Therefore, $\pi^{-1}(W^{\perp E/Z})\in \Gamma$ if and only if $\phi(W)^\perp \in \Gamma$. This is equivalent to $\phi(W)\in \Gamma^*$, and since $\phi(W)\leq Z^\perp$, we have $\phi(W)\in \Gamma^*|_{Z^\perp}$. Applying the identical argument to $\A$ shows that $\phi(W)\in \A^*|_{Z^\perp}$ for all $W\in (\A/Z)^*$. This establishes $(\S/Z)^*\simeq \S^*|_{Z^\perp}$. \qedhere}

\end{proof}

Finally, we introduce nomenclature to describe $q$-access structures as generalizations of definitions of classical access structures \cite{beimel2011secret}.
\begin{definition}
    Let $\S=(\Gamma,\A)$ be a $q$-access structure on $\L(E)$. 
    \begin{enumerate}
        \item $\S$ is \textit{degenerate} if $\Gamma=\varnothing$ or $\A=\varnothing$.
        \item The \textit{minimal reconstructing structure} of $\S$ is $\Gamma_{\min} = \{V\in \Gamma \,|\, \text{for all } W<V \Rightarrow W \notin \Gamma\}$.
        \item $\S$ is \textit{perfect} if $\overline{\Gamma} = \A$, or equivalently, $\Gamma \cup \A = \L(E)$.
        \item The \textit{minimum gap} when $\S$ is \c{ non-degenerate} is $$g(\S) = \min\{ \dim(V/W) \,|\, (V,W) \in \Gamma \times \c{\A}\colon W<V\}.$$
        \item $\mathbf{S}$ is \textit{connected} if for all $p\in \P(E)$ there exists $V\in \Gamma_{\min}$ such that $p\leq V$.
        \item $\S$ is $k$-\textit{threshold} if there exists $k\in \mathbb{N}$ such that $\Gamma = \{ V\leq E \,| \, \dim V \geq k\}.$
    \end{enumerate}
\end{definition}

 Suppose $\S=(\Gamma,\A)$ is a perfect $q$-access structure. Then the dual of $\S=(\Gamma,\A)$ is also perfect and satisfies $\S^*=(\overline{\Gamma^*},\overline{\A^*})$. Furthermore, $\S$ is entirely determined by $\Gamma_{\min}$, and all of its minors are perfect. The notion of perfect encapsulates the fact that either one obtains all information or no information; there is no partial information to be obtained. Non-degenerate perfect $q$-access structures have a minimum gap of $1$.

For the remainder of the section we only consider properties of $q$-access structures induced by generalized $q$-polymatroid ports. These results are the $q$-analogues of classical results pertaining to matroid theory and secret sharing \cite{cryptoeprint:2012/595}.

We first generalize the notion of information ratio to generalized $q$-polymatroid ports and demonstrate how it relates to the minimum gap. The information ratio is a measure of the size of each packet compared to the size of the secret, or message.

\begin{definition}
    Let $\mathbf{S} = \mathbf{S}_{P_0,P}(\mathcal{M})$ be a generalized $q$-polymatroid port. Then $$\sigma(\mathbf{S}) = \max_{p \in \mathcal{P}(P)}\{ \rho(p)\}/\rho(P_0)$$ is the \emph{information ratio} of $\mathbf{S}$.
\end{definition}

\begin{proposition} \label{prop:gap}
Let $\mathbf{S} = \mathbf{S}_{P_0,P}(\mathcal{M})=\c{(\Gamma,\A)}$ be a \c{non-degenerate} generalized $q$-polymatroid port. Then $\sigma(\mathbf{S})\geq g(\mathbf{S})^{-1}$.
\end{proposition}
\begin{proof}
Let $V\in \Gamma$ and $W\in \mathcal{A}$ with $W\leq V$ and fix some $Y\leq P$ such that $W\oplus Y = V$. Then
\[
    \rho(P_0) + \rho(W) = \rho(W+P_0) \leq \rho(V+P_0) = \rho(V) \leq \rho(W) + \rho(Y),
\]
from which it follows that $\rho(P_0)\leq \rho(Y) \leq \dim Y \max_{p\in \mathcal{P}(P)}\{\rho(p)\}$. The result then follows by taking the minimum over all such pairs $(V,W)$.
\end{proof}

Therefore, by Proposition \ref{prop:gap}, all generalized $q$-polymatroid ports with a minimum gap of $1$ have an information ratio of at least $1$. We thus conclude that if the size of the secret $\rho(P_0)$ is larger than the size of every share, then the port is not perfect.

\begin{corollary}
Let $\S=\S_{P_0,P}(\M)$ be a non-degenerate generalized $q$-polymatroid port. Suppose that $\sigma(\S) <1$. Then $\S$ is not perfect.
\end{corollary}

Since we have introduced the notions of contraction, restriction, duality, and equivalence for $q$-access structures, a natural question to ask is how these relate to generalized $q$-polymatroid ports and to the corresponding notions of the underlying $q$-polymatroid. In particular, the following result shows that contracting or restricting the $q$-access structure induced by a $q$-polymatroid is equivalent to first contracting or restricting the $q$-polymatroid itself, and then inducing a generalized $q$-polymatroid port.

\begin{proposition}\label{prop:minors}
    Let $\mathbf{S}_{P_0,P}(\mathcal{M})=\c{(\Gamma,\A)}$ be a generalized $q$-polymatroid port and let $Z\leq P$. Then $\mathbf{S}_{P_0,Z}(\mathcal{M}|_{P_0 + Z})$ is a generalized $q$-polymatroid port and
    $$\mathbf{S}_{P_0,P}(\mathcal{M})|_Z=\mathbf{S}_{P_0,Z}(\mathcal{M}|_{P_0 + Z}).$$
    Furthermore, if $Z\in \A$ then $\mathbf{S}_{\pi(P_0),\pi(P)}(\mathcal{M}/Z)$ is a generalized $q$-polymatroid port and
    $$\mathbf{S}_{P_0,P}(\mathcal{M})/Z = \mathbf{S}_{\pi(P_0),\pi(P)}(\mathcal{M}/Z),$$
    where $\pi\colon E\rightarrow E/Z$ is the canonical quotient map.
\end{proposition}
\begin{proof}
The statement about restrictions follows since $\rho|_{P_0+Z}(V) = \rho(V)$ for all $V\leq P_0+Z$. For contraction, we have $\rho_{E/Z}(\pi(P_0)) = \rho(P_0 \,|\, Z) = \rho(P_0)$ and $\rho_{E/Z}(\pi(P_0)\,|\, W) = \rho(P_0\,|\, \pi^{-1}(W))$ for all $W\leq P/Z$. Hence, $W$ belongs to the reconstructing, respectively privacy, structure of $\S_{\pi(P_0),\pi(P)}(\M/Z)$ if and only if $\pi^{-1}(W)$ belongs to $\Gamma$, respectively $\A$.
\end{proof}

Similarly, equivalent $q$-polymatroids yield equivalent ports under appropriate choices of $P_0$ and $P$, the proof of which follows directly by Definition \ref{def:equivalentmatroids}.

\begin{proposition} \label{prop:equivports}
    Let $\mathcal{M},\mathcal{M'}$ be a pair of equivalent $q$-polymatroids under the $\F_q$-isomorphism $\phi$.
    Let $\mathbf{S}_{P_0,P}(\mathcal{M})$ be a generalized $q$-polymatroid port. Then $\mathbf{S}_{\phi(P_0),\phi(P)}(\mathcal{M}')$ is a generalized $q$-polymatroid port and $$\mathbf{S}_{P_0,P}(\mathcal{M}) \simeq \mathbf{S}_{\phi(P_0),\phi(P)}(\mathcal{M}').$$
\end{proposition}

Dualizing the port or the underlying $q$-polymatroid yields equivalent ports in the following sense.

\begin{proposition}
    Let $\mathbf{S}_{P_0,P}(\mathcal{M})=\c{(\Gamma,\A)}$ be a non-degenerate generalized $q$-polymatroid port with $\rho(P_0) = \dim P_0$. Then $\mathbf{S}_{P^{\perp},P_0\textcolor{white}{)}\hspace{-1.5mm}^{\perp}}(\mathcal{M}^*)$ is a generalized $q$-polymatroid port and
    \[
        \mathbf{S}_{P_0,P}(\mathcal{M})^* \simeq \mathbf{S}_{P^{\perp},P_0\textcolor{white}{)}\hspace{-1.5mm}^{\perp}}(\mathcal{M}^*).
    \]
\end{proposition}
\begin{proof}
    Let $V\leq P$ and $\bot = \bot E$. Then
    \[
        \rho^*(P^\perp\,|\, V^\perp \cap P_0^\perp) = \dim P^\perp - \rho(P_0\,|\, V).
    \]
    This shows that $\mathbf{S}_{P^{\perp},P_0\textcolor{white}{)}\hspace{-1.5mm}^{\perp}}(\mathcal{M}^*)$ is a generalized $q$-polymatroid port as $\mathbf{S}_{P_0,P}(\mathcal{M})$ is non-degenerate, which yields $\rho^*(P^\perp) = \dim P^\perp$.

    Let $\tau \colon \mathcal{L}(P)\rightarrow \mathcal{L}(P_0^\perp)$ be given by $V\mapsto V^\perp\cap P_0^\perp$. Since $P_0\oplus P=E$, the bilinear form on $E$ induces a non-degenerate pairing between $P$ and $P_0^\perp$, and hence $\tau$ is an anti-monotone bijection. Let $\perp P$ denote the orthogonal complement in $P$ with respect to the non-degenerate symmetric bilinear form used to define $\mathbf{S}_{P_0,P}(\mathcal{M})^*$. Then the map $V\mapsto \tau(V^{\perp P})$ is induced by an $\F_q$-linear isomorphism $\varphi\colon P\rightarrow P_0^\perp$. Therefore, the result follows as
    \[
        V\in \Gamma^* \Leftrightarrow V^{\perp P}\in \Gamma
        \Leftrightarrow \rho^*(P^\perp\,|\, \varphi(V)) = \rho^*(P^\perp)
    \]
    and
    \[
        W\in \mathcal{A}^* \Leftrightarrow W^{\perp P}\in \mathcal{A}
        \Leftrightarrow \rho^*(P^\perp\,|\, \varphi(W)) =0.
    \qedhere
    \]
\end{proof}

An access structure realized by a $q$-matroid port is perfect, and its minimal access structure can be characterised by the minors of the underlying $q$-matroid.
\begin{theorem}\label{thm:qmatroidports} Let $\mathbf{S}_{P_0,P}(\mathcal{M})=\c{(\Gamma,\A)}$ be a $q$-matroid port. Then $\mathbf{S}_{P_0,P}(\mathcal{M})$ is perfect and \begin{align*}
    \Gamma_{\min} = \{V\leq P \,|\, V \text{ is a basis of } \M|_{V+P_0} \text{ and } W+P_0 \text{ is independent in }\M\text{ for all } W<V\}.
    \end{align*}
\end{theorem}
\begin{proof}
$\mathbf{S}_{P_0,P}(\mathcal{M})$ is perfect as $\rho$ is integer-valued and $0\leq \rho(P_0\,|\,V)\leq \rho(P_0) =1$ for any $V\leq P$. 
If $V \in \Gamma_{\min}$ then $\rho(P_0+V)=\rho(V)$ and 
$\rho(P_0+W)=\rho(P_0)+\rho(W)=1+\rho(W)$ for every $W \lneq V$. 
Suppose that $V$ is a dependent space of $\M$. Then there exists $W \lneq V$ such that $\rho(W) = \rho(V)$.
Then $\rho(V+P_0)=\rho(V)=\rho(W) = \rho(W+P_0)-1 \leq \rho(V+P_0)-1$, which is clearly impossible. It follows that $V$ is a basis of $\M|_{V+P_0}$. Furthermore, we have that $W$ is independent in $\M$ for any $W \lneq V$, which yields that 
$\dim(W+P_0) = \dim(W)+1=\rho(W)+1 =\rho(W+P_0)$, so that $W+P_0$ is an independent space of $\M$.  

For the converse, consider $V\leq P$ such that it is a basis of $\M|_{V+P_0}$. Then $\rho(V) = \rho(V+P_0)$, so $V\in \Gamma$. If every $W<V$ has $W+P_0$ independent in $\M$, then $\rho(P_0\,|\,W)=1=\rho(P_0)$, so $W\in \A$. Thus $V\in \Gamma_{\min}$.
\end{proof}

\begin{remark}
Whereas any linear code induces a matroid, an ${\mathbb F}_q$-linear rank-metric code does not necessarily induce a $q$-matroid, but rather a $q$-polymatroid. Consequently, not every port induced by a rank-metric code is perfect. Theorem \ref{thm:qmatroidports} is the generalization of a well-known result on matroid ports, see \cite{farre_padro2007}. However, it differs in the sense that the minimal access structure of a $q$-matroid port is not the set of $V\leq P$ for which $V+P_0$ is a circuit in the underlying $q$-matroid. While this condition is sufficient, Example \ref{exmp:circuitminimal} shows it is not necessary. All examples in this paper were computed using the Magma algebra system \cite{MR1484478}. \label{remark1}
\end{remark}

\begin{example} \label{exmp:circuitminimal}
Consider $\mathcal{C} = \langle M_1,M_2,M_3,M_4\rangle_{\mathbb{F}_2}\leq \F_2^{4\times 2}$, where
\[ M_1 =
\begin{bsmallmatrix}
    1 & 0 \\
    0 & 0 \\
    1 & 1 \\
    0 & 1 \\
\end{bsmallmatrix},\:
M_2 =
\begin{bsmallmatrix}
    0 & 1 \\
    0 & 0 \\
    1 & 0 \\
    1 & 1 \\
\end{bsmallmatrix},\:
M_3 =
\begin{bsmallmatrix}
    0 & 0 \\
    1 & 0 \\
    1 & 0 \\
    1 & 0 \\
\end{bsmallmatrix},\:
M_4 =
\begin{bsmallmatrix}
    0 & 0 \\
    0 & 1 \\
    0 & 1 \\
    0 & 1 \\
\end{bsmallmatrix}.
\]
In fact $\C$ may be viewed as a vector rank-metric code over $\F_4$, generated by $m_1 = (1,0,\alpha^2,\alpha)$ and $m_2=(0,1,1,1)$ for a primitive element $\alpha \in \F_4$. Hence, $\M_\C$ is a $q$-matroid.
Let $P_0 = \langle \mathbf{e}_1 \rangle_{\mathbb{F}_2}$ and $P = \langle \mathbf{e}_2,\mathbf{e}_3,\mathbf{e}_4\rangle_{\mathbb{F}_2} $. Then $\rho_{\mathcal{C}}(P_0)>0$ so $\mathbf{S}_{P_0,P}(\mathcal{M}_{\mathcal{C}})$ is a perfect $q$-matroid port by Theorem \ref{thm:qmatroidports}. Furthermore, 
\[
    \Gamma_{\min} = \{\langle  \mathbf{e}_2+\mathbf{e}_4 \rangle_{\mathbb{F}_2},
    \langle  \mathbf{e}_2+\mathbf{e}_3 \rangle_{\mathbb{F}_2},
\langle  \mathbf{e}_3+\mathbf{e}_4 \rangle_{\mathbb{F}_2}
\}.
\]
Note, however, that the circuits of $\mathcal{M}_{\mathcal{C}}$ are $\{\langle\mathbf{e}_1, \mathbf{e}_2+\mathbf{e}_3 \rangle_{\mathbb{F}_2},\langle \mathbf{e}_2+\mathbf{e}_4,\mathbf{e}_3+\mathbf{e}_4 \rangle_{\mathbb{F}_2},$ $\langle\mathbf{e}_1,\mathbf{e}_2+\mathbf{e}_4 \rangle_{\mathbb{F}_2}$, $\langle \mathbf{e}_1+\mathbf{e}_2+\mathbf{e}_4,\mathbf{e}_3+\mathbf{e}_4 \rangle_{\mathbb{F}_2},\langle\mathbf{e}_1+\mathbf{e}_3+\mathbf{e}_4 \rangle_{\mathbb{F}_2}\}$, while $v =\langle \mathbf{e}_3+\mathbf{e}_4 \rangle_{\mathbb{F}_2} + P_0$ is not a circuit as $\langle \mathbf{e}_1+\mathbf{e}_3+\mathbf{e}_4 \rangle_{\mathbb{F}_2}\leq v$.
\end{example}

A codeword of a linear code is called \emph{minimal} if its support does not contain the support of any other linearly independent codeword. This notion was extended to vector rank-metric codes in~\cite{ALFARANO2022105658}. In the setting of classical secret sharing based on linear codes, Massey~\cite{Massey1999MinimalCA} showed that the minimal access structure of a representable matroid port is completely determined by certain minimal codewords of the dual code. Here we establish a $q$-analogue of the sufficient direction of this result for vector rank-metric codes. As in Remark~\ref{remark1}, minimal dual codewords give rise to minimal reconstructing spaces, but the converse need not hold, even when the dual code is minimal, as illustrated in Example~\ref{exmp:minimalcodewords}. Recall that for an element $v \in \mathbb{F}_{q^{m}}^{n}$, we write $\Pi(v)$ for the matrix in $\mathbb{F}_{q}^{n \times m}$ obtained by expanding the entries of $v$ with respect to the basis $\Pi = \{\gamma_{1}, \ldots, \gamma_{m}\}$ of $\mathbb{F}_{q^{m}}$ over $\mathbb{F}_{q}$.

\begin{definition}\c{
Let $\C\leq \F_{q^m}^n$ be a vector rank-metric code and $\Pi$ a basis of $\F_{q^m}$ over $\F_q$. A codeword $v\in \C$ is a \emph{minimal codeword} if every $v'\in \C$ with $\colsp(\Pi(v')) \leq \colsp(\Pi(v))$ is of the form $v'=\alpha v$ for some $\alpha \in \F_{q^m}$. The set of minimal codewords of $\C$ is denoted by $\C_{\min}$. Furthermore, $\C$ is \emph{minimal} if $\C_{\min}=\C$. \label{def:rankmetricminimal}}
\end{definition}

\c{For the proof of the following result we will require the notion of orthogonal bases. Let $\tr_{\F_{q^m}/\F_q}\colon \F_{q^m}\rightarrow \F_q$ denote the field trace map defined by $\tr_{\F_{q^m}/\F_q}(\alpha) = \alpha + \alpha^q + \ldots + \alpha^{q^{m-1}}$. Two bases $\Pi = \{\gamma_1,\ldots,\gamma_m\}$ and $\Pi' = \{\gamma_1',\ldots,\gamma_m'\}$ of $\F_{q^m}$ over $\F_q$ are called \emph{orthogonal} if $$\tr_{\F_{q^m}/\F_q}(\gamma_i \gamma_j') =\begin{cases}
    0, & \text{if }i\neq j,\\
    1, & \text{if } i =j.
\end{cases}$$
For any basis $\Pi$ of $\F_{q^m}$ over $\F_q$ there exists a unique orthogonal basis $\Pi'$ \cite{Lidl_Niederreiter_1996}.}

\begin{theorem}[\cite{rav_2016}, Theorem 21]
\label{thm:dualcode} Let $\C\leq \F_{q^m}^n$ be a vector rank-metric code, and let $\Pi,\Pi'$ be orthogonal bases of $\F_{q^m}$ over $\F_q$. Then $\Pi'(\C^\perp) = \Pi(\C)^\perp$.
\end{theorem}

\begin{theorem}\label{thm:minimal}
Let $\C\leq \F_{q^m}^n$ be a vector rank-metric code and $\Pi$ be any basis of $\F_{q^m}$ over $\F_q$. Consider the $q$-matroid port $\S_{P_0,P}(\M_{\Pi(\C)})=\c{(\Gamma,\A)}$. If there exists $X\in (\C^\perp)_{\min}$ such that $P_0\leq \colsp( \Pi(X))$, then $\colsp(\Pi(X)) \cap P \in \Gamma_{\min}$.
\end{theorem}
\begin{proof}
\c{Denote $\D = \Pi(\C)$, and let $\Pi'$ be the orthogonal basis of $\Pi$ over $\F_q$. Then $\D^\perp = \Pi(\C)^\perp = \Pi'(\C^\perp)$ by Theorem~\ref{thm:dualcode}. Recall that for any $v \in \F_{q^m}^n$, $\colsp(\Pi(v)) = \colsp(\Pi'(v))$.}

Suppose $X\in (\C^\perp)_{\min}$ such that $P_0\leq \colsp( \Pi(X))$ and let $V=\colsp(\Pi(X))=\c{\colsp(\Pi'(X))}$. We then have 
$P_0+V\cap P = (P_0+P) \cap V = V$ as $P_0\oplus P = \F_q^n$. Since $X\in (\C^\perp)_{\min}$, \c{its representation $\Pi'(X)$ is minimal in $\D^\perp$}, which yields $\dim \D^\perp(V\cap P) = 0$, so by Lemma \ref{lem:ccdual} and the fact that $\dim(V \cap P)= \dim(V)-1$, we obtain
\[
    \dim \D^\perp(V\cap P) = mn - \dim \D - m(n-\dim V +1) + \dim \D((V\cap P)^\perp) = 0.
\]
By Lemma \ref{lem:ccdual} and since $\dim \D^\perp(V) = m$, we have that
\begin{align*}
    \dim \D((P_0 + V\cap P)^\perp) &= \dim \D(V^\perp) \\ 
    &= \dim \D^\perp(V) - mn + \dim \D + m\dim V^\perp \\
    &= -mn + \dim \D + m(n-\dim V +1),
\end{align*}
and
\begin{align*}
    \dim \D((V\cap P)^\perp) &= \dim \D^\perp(V\cap  P) - mn + \dim \D + m\dim( (V\cap P)^\perp) \\
    &= - mn + \dim \D + m(n-\dim V+1).
\end{align*}
Therefore, $\rho_{\D}(P_0\,|\, V\cap P) = m^{-1}( \dim \D((P_0 + V\cap P)^\perp)-\dim \D((V\cap P)^\perp))=0$, so $V\cap P \in \Gamma$. 

Suppose now that $W<V\cap P$, so $W+ P_0 < V$, which implies $\dim \D^\perp (W+P_0) = \dim \D^\perp(W) = 0$. Again by Lemma \ref{lem:ccdual}, we get
\begin{align*}
   \dim \D((W+P_0)^\perp) &= \dim \D^\perp(W+P_0) -mn + \dim \D + m \dim ((W+P_0)^\perp) \\&= 
    -mn + \dim \D + m (n- (\dim W + 1))
\end{align*}
and
\begin{align*}
    \dim \D(W^\perp) &=\dim  \D^\perp(W) - mn + \dim \D + m(n-\dim W)\\
    &= -mn + \dim \D + m(n-\dim W),
\end{align*}
which proves $\rho_\D(P_0\,|\, W) = \rho_\D(P_0)$, so $V\cap P\in \Gamma_{\min}$.
\end{proof}

\begin{example}
    \label{exmp:minimalcodewords}
    Let $\F_8 = \F_2[\alpha]$, where $\alpha^3+\alpha+1=0$. Consider the vector rank-metric code $\C$ generated by
    \[
    G =
    \begin{bmatrix}
        1 & 0 & \alpha^2 & \alpha^3 \\
        0 & 1 & 1 & \alpha^2 \\
    \end{bmatrix}.
    \]
A direct computation shows that every non-zero codeword of $\C^\perp$ is minimal, and hence $(\C^\perp)_{\min}=\C^\perp$. Let $P_0 = \langle \mathbf{e}_1\rangle_{\F_2}$ and $P = \langle \mathbf{e}_2,\mathbf{e}_3,\mathbf{e}_4\rangle_{\F_2}$. Then $\S_{P_0,P}(\M_{\Pi(\C)})$ is a $q$-matroid port with minimal access structure
\[
    \Gamma_{\min} = \{ \langle\mathbf{e}_2+\mathbf{e}_3\rangle_{\F_2}, \langle \mathbf{e}_2,\mathbf{e}_4\rangle_{\F_2},\langle \mathbf{e}_2,\mathbf{e}_3+\mathbf{e}_4\rangle_{\F_2}, \langle \mathbf{e}_2+\mathbf{e}_4,\mathbf{e}_3\rangle_{\F_2}, \langle \mathbf{e}_3,\mathbf{e}_4\rangle_{\F_2}\}.
\]
Additionally, 
\[
    \{\colsp(\Pi(X)) \cap P \mid X\in (\C^\perp)_{\min},\ P_0 \leq \colsp(\Pi(X))\} = \{\langle\mathbf{e}_2+\mathbf{e}_3\rangle_{\F_2},\langle \mathbf{e}_2,\mathbf{e}_4\rangle_{\F_2},\langle \mathbf{e}_2,\mathbf{e}_3+\mathbf{e}_4\rangle_{\F_2} \},
\]

which is a proper subset of $\Gamma_{\min}$. Thus, the sufficient condition of Theorem~\ref{thm:minimal} is not necessary, even when $\C^\perp$ is a minimal code.
\end{example}

Suppose that one creates a coset coding scheme using a vector rank-metric code $\C\leq \F_{q^m}^n$, where the nested pair of coset codes are $\Pi(\C)(P_0^\perp) \leq \Pi(\C)$, with $\dim P_0=1$. Theorem~\ref{thm:minimal} shows that minimal codewords of $\C^\perp$ whose rank support contains $P_0$ give rise to elements of the minimal access structure. However, Example~\ref{exmp:minimalcodewords} shows that these codewords need not determine the entire minimal access structure, even when $\C^\perp$ is a minimal code.

Just as uniform matroids correspond to threshold access structures in classical secret sharing via MDS codes, the $q$-analogue holds in the rank-metric setting. Specifically, uniform $q$-matroids, which arise from MRD codes, also yield threshold $q$-access structures. This parallel underscores the structural consistency between the classical and $q$-analogue frameworks and further highlights the role of MRD codes as rank-metric analogues of MDS codes in secret sharing.

\begin{proposition} Suppose $\S = \S_{P_0,P}(\mathcal{U}_{\dim E,k})$ is a generalized $q$-matroid port. Then $\S$ is $k$-threshold.\label{prop:uniform}
\end{proposition}
\begin{proof}
The rank function of $\mathcal{U}_{n,k}$ satisfies $\rho(V) = \min\{\dim V, k\}$ for all $V\leq E$, so for any $V\leq P$, $\rho(P_0+V) = \rho(V)$ if and only if $\dim V \geq k$.
\end{proof}
Note for Proposition \ref{prop:uniform} that if $\dim P <k$ then $\S_{P_0,P}(\mathcal{U}_{n,k})$ is degenerate.

We next establish a $q$-analogue of the Brickell--Davenport theorem, showing that when a port of a $q$-polymatroid satisfies certain ideality and connectivity conditions, the restricted structure $\M|_P$ is in fact a $q$-matroid. We first prove some preliminary lemmas.

\begin{lemma} Let $\S_{P_0,P}(\mathcal{M})$ be a perfect and connected $q$-polymatroid port such that $\rho(P_0)= \dim P_0 = 1$. For any non-zero $V\in \A$ there exists $W\in \A$ such that $V\cap W=0$ and $V+W\in \Gamma$, with the additional minimality property that if $W'<W$ then $V+W'\in \A$.
\label{lem:brickell1}
\end{lemma}
\begin{proof}
Fix $V\in \A$ and choose any $p\in \P(V)$. By connectedness of the port, there exists $G\in \Gamma_{\min}$ with $p\leq G$. In particular we have $G\cap V \geq p>0$. Choose now a complement $U< G$ of $G\cap V$ such that $U\oplus (G\cap V) = G$. $U$ is a proper subspace of $G$, so $U\notin\Gamma$ by minimality of $G$. By perfectness, $U\in\A$ and $U\cap V=0$. As $G\leq V+U$, we have by monotonicity of $\Gamma$ that $V+U\in\Gamma$. Consider now the non-empty family $\mathcal S=\{U'\leq U\mid V+U'\in\Gamma\}$, as $U\in\mathcal S$. There exists a minimal element of $\mathcal S$, which we denote by $W$. Clearly $W\leq U<G$, and $G\in\Gamma_{\min}$, so $W\in\A$. Additionally, by minimality of $W$ in $\mathcal S$, if $W'<W$ then $W'\notin\mathcal S$, so $V+W'\notin\Gamma$, which implies $V+W'\in\A$.
\end{proof}

\begin{lemma}
Let $\S_{P_0,P}(\mathcal{M})$ be a perfect $q$-polymatroid port such that $\rho(P_0) = \dim P_0 =1$. Suppose $W\leq P$ and $y\in \P(P)$ such that both $W\in \A$ and $y+W\in \Gamma$. Then \label{lem:brickell2}
$$
    \rho(y\,|\, W) = 1 \quad \text{and} \quad \rho(y\,|\, W+P_0)=0.
$$
\end{lemma}
\begin{proof} Note that
\[
    \rho(W) + \rho(P_0\,|\, W) + \rho(y\,|\, W+P_0) = \rho(y+W+P_0) = \rho(y+W) + \rho(P_0 \,|\, y+W).
\]
By submodularity and non-negativity of $\rho$ we have that
\[
    1\leq \rho(y\,|\, W+P_0) + 1 \leq  \rho(y\,|\, W) \leq 1,
\]
and the result follows.
\end{proof}
\begin{definition}
A generalized $q$-polymatroid port $\S_{P_0,P}(\M)$ is called \textit{ideal} if $\rho(P_0)= \rho(p)$ for all $p\in \P(P)$.
\end{definition} 
In terms of a coset coding scheme, ideal generalized $q$-polymatroid ports can be seen as those for which the size of the secret is the same as the information in each packet.

\begin{theorem}[$q$-Brickell-Davenport Theorem] Let $\S_{P_0,P}(\M)$ be an ideal, perfect and connected $q$-polymatroid port with $\rho(P_0) = \dim P_0 = 1$. Then $\M|_P$ is a $q$-matroid. \label{thm:brickelldavenport}
\end{theorem}
\begin{proof}
Towards a contradiction, suppose that there exists a subspace $V\leq P$ such that $\rho(V)$ is not a non-negative integer. Assume that with respect to inclusion, $V$ is a minimal subspace of $P$ with this property. Write $s= \lfloor \rho(V) \rfloor$, so $s < \rho(V) <s+1.$ For any $U\leq P$ and any $u\in \P(U)$ let $U_u$ denote a hyperplane of $U$ that is a complement of $u$ in $U$, i.e., $U = U_u \oplus u$. Let $y\in \P(V)$, so $\rho(V_y)\in \mathbb{N}_0$, and as $\rho(y) = 1$ it follows that
\[
    \rho(V_y) \leq s <\rho(V) = \rho(V_y + y) \leq \rho(V_y) +1.
\]
This gives us that $s-1 <\rho(V_y) \leq s$ and since $\rho(V_y)$ is an integer we obtain $\rho(V_y) = s$. In particular,
\[
    0<\rho(y\,|\, V_y) <1.
\]
Since $\S_{P_0,P}$ is perfect, $\Gamma$ and $\A$ form a partition of $\L(P)$, and as such we consider the case that $V \in \Gamma$ and the case $V \in \A$.

Suppose first $V\in \Gamma$. For any $y\in \P(V)$, we have $V_y\in \Gamma$ by Lemma \ref{lem:brickell2}, since otherwise $V_y\in \A$ and Lemma \ref{lem:brickell2} would give $\rho(y\,|\,V_y)=1$, contradicting $0<\rho(y\,|\,V_y)<1$. Suppose now $V'\in \Gamma_{\min}$ such that $V'\leq V$ and let $y'\in \P(V')$. Then $V'_{y'}\in \A$. Choose $V_{y'}$ such that $V'_{y'}\leq V_{y'}$. By Lemma \ref{lem:brickell2}, we have $\rho(y'\,|\, V'_{y'}+P_0) = 0$, and hence by submodularity,
\[
    0\leq \rho(y'\,|\, V_{y'}+P_0)\leq \rho(y'\,|\, V'_{y'}+P_0)=0.
\]
Since $V\in\Gamma$, we then obtain
\[
    \rho(y'\,|\, V_{y'}) = \rho(V) - \rho(V_{y'}) = \rho(V+P_0) - \rho(V_{y'}) = \rho(y'\,|\, V_{y'}+P_0)=0,
\]
which contradicts $0<\rho(y'\,|\,V_{y'})<1$.

Suppose now $V\in \A$. By Lemma \ref{lem:brickell1} there exists $W\in \A$ such that $V+W \in \Gamma$, $V\cap W = 0$ and if $W'<W$ then $V +W'\in \A$. Let $\dim W =\ell$ and choose $x_1,\dots,x_\ell \in \P(W)$ such that
\[
    W = x_1 \oplus x_2 \oplus \cdots \oplus x_\ell.
\]
For each $j\in [\ell]$ let $W_j = \oplus^\ell_{i=j}x_i$, and set $W_{\ell+1}=0$. For each $j\in[\ell]$, choose the hyperplane $W_{x_j}$ of $W$ complementary to $x_j$ such that $W_{j+1}\leq W_{x_j}$. Since $W_{x_j}<W$, we have $V+W_{x_j}\in\A$, while $V+W\in\Gamma$. Hence, by Lemma \ref{lem:brickell2}, $\rho(x_j\,|\,V+W_{x_j})=1$. Since $W_{j+1}\leq W_{x_j}$, submodularity gives
\[
    1=\rho(x_j\,|\,V+W_{x_j})\leq \rho(x_j\,|\,V+W_{j+1})\leq \rho(x_j)=1.
\]
Thus, $\rho(x_j\,|\,V+W_{j+1})=1$ for all $j=1,\ldots,\ell$. We now have that
\[
\begin{aligned}
    \rho(V+W)   &= \rho(V+W_2)+\rho(x_1\,|\,V+W_2) \\
                &= \rho(V+W_3)+\rho(x_2\,|\,V+W_3)+\rho(x_1\,|\,V+W_2) \\
                &\: \,\,\vdots \\
                &= \rho(V)+\sum^\ell_{i=1}\rho(x_i\,|\,V+W_{i+1}) \\
                &= \rho(V)+\dim W.
\end{aligned}
\]

Now, suppose $x\in \P(V)$. We have $V_x+W\in \Gamma$ for all $x\in \P(V)$. Indeed, if $V_x+W\in\A$, then since $V+W\in\Gamma$, Lemma \ref{lem:brickell2} would give $\rho(x\,|\,V_x+W)=1$, contradicting
\[
    \rho(x\,|\,V_x+W)\leq \rho(x\,|\,V_x)<1.
\]
Furthermore, we have
\[
    \rho(V_x+W)=\rho(V_x)+\dim W.
\]
To see this, submodularity gives $\rho(V_x+W)\leq\rho(V_x)+\dim W$. On the other hand, since $\rho(x\,|\,V_x+W)\leq\rho(x\,|\,V_x)$ and $\rho(V+W)=\rho(V)+\dim W$, we obtain
\[
\begin{aligned}
    \rho(V_x+W)
    &= \rho(V+W)-\rho(x\,|\,V_x+W) \\
    &\geq \rho(V)+\dim W-\rho(x\,|\,V_x) \\
    &= \rho(V_x)+\dim W,
\end{aligned}
\]
and hence equality follows.

Choose $Z\in \Gamma$ that is minimal among those that contain $W$ and are contained in $V+W$. Thus, $W<Z \leq V+W$ and $V\cap Z >0$. Let $z\in \P(V\cap Z)$ and choose the hyperplane $Z_z$ of $Z$ complementary to $z$ such that $W\leq Z_z$. By the minimality of $Z$, we have $Z_z\notin\Gamma$, and hence $Z_z\in\A$. Lemma \ref{lem:brickell2} therefore gives $\rho(z\,|\,Z_z+P_0)=0$.

Choose $V_z$ such that $Z_z\cap V\leq V_z$. Since $Z_z\leq V+W$ and $W\leq Z_z$, we have $Z_z\leq V_z+W$. Hence, by submodularity,
\[
    0\leq \rho(z\,|\,V_z+W+P_0)\leq \rho(z\,|\,Z_z+P_0)=0.
\]
Moreover, both $V_z+W$ and $V+W$ belong to $\Gamma$, so adjoining $P_0$ does not change their ranks. It follows that
\[
    \rho(z\,|\,V_z+W)=\rho(z\,|\,V_z+W+P_0)=0.
\]
Therefore,
\[
    \rho(V) + \dim W = \rho(V+W) = \rho(V_z+W) + \rho(z\,|\, V_z+W) = \rho(V_z+W) = \rho(V_z) + \dim W,
\]
which implies $\rho(V)=\rho(V_z)$. This contradicts $\rho(V_z)\leq s<\rho(V)$, which concludes the proof.
\end{proof}

\begin{corollary}
Let $\M$ be a $q$-polymatroid and let $P_0\in \P(E)$ such that $\rho(P_0)=1$. If $\S_{P_0,P}(\M)$ is an ideal, perfect, and connected $q$-polymatroid port for every complement $P$ of $P_0$ then $\M$ is a $q$-matroid. \label{cor:brickelldavenport}
\end{corollary}
\begin{proof}
For any $V,W\leq P$ let $\L(V)+\L(W) \coloneqq \{ V'+W' \,|\, V'\leq V, W'\leq W\}$. Denote the complements of $P_0$ in $E$ as $\mathbf{C}(P_0) = \{P\leq E \,|\, P\oplus P_0 = E\}$, so
\[
    \L(E) = \bigcup_{P\in \mathbf{C}(P_0)} (\L(P) + \L(P_0)).
\]
Theorem~\ref{thm:brickelldavenport} yields integrality on $\L(P)$, and for any $V\leq P$, $V+P_0$ also has integral rank as $\rho(V+P_0) = \rho(V) + \rho(P_0\,|\,V)$, where the latter summand is 0 or 1 by Theorem~\ref{thm:qmatroidports}.
\end{proof}

{
The requirements of Corollary \ref{cor:brickelldavenport} are restrictive, since a given $P_0 \in \mathcal{P}(E)$ has $q^{n-1}$ complements. Furthermore, the conditions of Theorem \ref{thm:brickelldavenport} do not guarantee the entire $q$-polymatroid is a $q$-matroid, making it of interest to identify relaxed conditions. With regard to classical secret sharing, in \cite{farre_padro2007} the authors generalized the Brickell-Davenport theorem to schemes with information rates above $2/3$ by utilizing the independent sequence method and Seymour's forbidden minor characterization \cite{10.1093/qmath/27.4.407}. Adapting these techniques to the rank-metric setting is a natural question, but would require the development of further structural tools.
}

\begin{remark}
The results of this section pertaining to general $q$-polymatroids extend naturally to finite modular, complemented lattices which include finite Boolean and subspace lattices. In this sense, the classical theory of access structures and polymatroid ports is recovered as a special case of this more general setting.
\end{remark}

\section{Conclusions and future work}
\label{sec:conclusions}

In this work, we developed a $q$-polymatroidal framework for analyzing information leakage in secure linear network coding schemes based on nested rank-metric codes. We showed that the information gained by a wiretapper observing an arbitrary collection of network links is fully characterized by the conditional rank function of a representable $q$-polymatroid associated with the underlying code pair. This provides a direct combinatorial interpretation of information leakage that is independent of the specific network topology and depends only on the combinatorial structure of the code.

Building on this connection, we introduced the notions of $q$-polymatroid ports and $q$-access structures, which generalize matroid ports and access structures from the set-based setting to the subspace-based setting. These objects capture reconstructing spaces and privacy spaces in a unified manner and naturally extend several well-known correspondences from linear secret sharing. In particular, we established a rank-metric analogue of the sufficient direction of Massey’s correspondence, showing that certain minimal codewords of the dual code give rise to minimal reconstructing spaces, while the converse need not hold, even when the dual code is minimal. We also proved a $q$-analogue of the Brickell--Davenport theorem, showing that for an ideal, perfect, and connected $q$-polymatroid port, the restriction of the underlying $q$-polymatroid to the participant space is a $q$-matroid under the stated structural assumptions.

By formulating information leakage in terms of $q$-polymatroid ports, we obtain a setting in which duality, minors, and minimal reconstructing spaces can be analyzed using standard $q$-polymatroid techniques. This makes it possible to study secrecy properties of rank-metric schemes using the same structural tools that are familiar from matroid-based secret sharing.

The framework developed in this paper suggests several directions for further research. From a combinatorial perspective, it is natural to further investigate the structure theory of $q$-access structures independently of representability. In matroid theory, excluded-minor characterizations and connectivity properties play a central role in understanding which access structures arise from linear secret sharing schemes. Developing analogous characterizations for $q$-access structures may shed light on which secrecy patterns can be realized by rank-metric codes. This could open up a new line of research trying to find and construct such specific rank-metric codes. 

Another direction concerns the classification of ideal and threshold $q$-access structures. While MRD codes give rise to perfect threshold structures, it remains open to characterize all $q$-matroids whose ports induce threshold or near-threshold behavior. Such results would parallel classical theorems on uniform matroids and ideal secret sharing and could inform the design of rank-metric schemes with optimal trade-offs between secrecy and efficiency. From the coding-theoretic viewpoint, extending the present framework to broader classes of codes is of interest. In particular, it would be valuable to study whether similar $q$-polymatroidal descriptions arise for non-linear rank-metric codes, subspace codes, or codes over more general algebraic structures.

Finally, the entropy-based characterization of representable $q$-polymatroids suggests deeper connections between information theory and $q$-polymatroid theory. Exploring whether general $q$-polymatroids admit operational interpretations in terms of entropy or mutual information, beyond the representable case, may provide new insights into abstract information inequalities in the subspace setting.

{
\section*{Acknowledgements}
The authors are very grateful for the comments made by anonymous reviewers during previous drafts of this manuscript, which greatly improved the presentation of these results.
}

\bibliographystyle{IEEEtran}
{\Large \bibliography{bibliography}}

\end{document}